\documentclass[12pt,a4paper]{scrartcl}
\usepackage[ngerman,english]{babel}
\usepackage{setspace}
\usepackage{lscape}
\usepackage{amsmath}
\usepackage{amssymb}
\usepackage{graphicx}
\usepackage{amsthm}
\usepackage{mathrsfs}
\usepackage{filecontents}
\usepackage{hyperref}
\usepackage{color}
\usepackage[top=2.75cm,bottom=2.50cm,left=3.00cm,right=2.50cm]{geometry}
\usepackage{mathtools}
 \usepackage{bbm}
\usepackage[T1]{fontenc}
\usepackage{multirow}
\usepackage{colonequals}
\usepackage{ wasysym }
\usepackage{authblk}

\newcommand{\defeq}{\vcentcolon=}
\newcommand{\eqdef}{=\vcentcolon}
   
\newcommand{\R}{\mathbb R}

\theoremstyle{plain}
\newtheorem{thm}{Theorem}
\newtheorem{prop}[thm]{Proposition}
\newtheorem{cor}[thm]{Corollary}
\newtheorem{lem}[thm]{Lemma}

\theoremstyle{definition}
\newtheorem{defi}[thm]{Definition}
\newtheorem{nota}[thm]{Notation}

\newtheorem{rem}[thm]{Remark}
\theoremstyle{definition}

\title{Photon sphere uniqueness in higher-dimensional electrovacuum spacetimes}
\author[1]{Sophia Jahns}
\affil[1]{Department of Mathematics, University of T\"ubingen\\ 

Auf der Morgenstelle 10, 72076 T\"ubingen, Germany}
\date{}     
\begin{document}
\maketitle

\section{Abstract}
We show a uniqueness result for the $n$-dimensional spatial Reissner--Nordstr\"om ma\-ni\-fold: a static, electrovacuum, asymptotically flat system which is asymptotically Reissner--Nordstr\"om 
is a subextremal Reissner--Nordstr\"om manifold with positive mass, provided that its inner boundary is a (possibly disconnected) photon sphere that fulfils a suitably defined quasilocal subextremality condition.

Our result implies a number of earlier uniqueness results for the Schwarzschild and the Reissner--Nordstr\"om manifolds in the static, (electro-)vacuum, asymptotically flat context, 
both for photon sphere and black hole inner boundaries, in the tradition of Bunting--Masood-ul Alaam~\cite{bunting} and Ruback~\cite{ruback}. 
The proof relies on the ideas from those works, combined with newer techniques developed by Cederbaum--Galloway~\cite{cedergal} and Cederbaum~\cite{carlahigh}. 

\section{Introduction}

The $n+1$-dimensional Reissner--Nordstr\"om spacetimes are a 2-parameter family (labelled with a \emph{mass} $m$ and a \emph{charge} $q$) 
of static, spherically symmetric, electrically charged, asymptotically flat solutions to the Einstein equations. A Reissner--Nordstr\"om spacetime is called \emph{subextremal} (\emph{extremal, superextremal}) if $m^2>q^2$ (if $m^2=q^2$, $m^2<q^2$). 

A subextremal or extremal Reissner--Nordstr\"om spacetime contains a unique photon sphere, on which light can get trapped (for a precise definition of photon spheres, see Definition~\ref{photonspheredefi}), 
while a superextremal Reissner--Nordstr\"om spacetime may contain two photon spheres or none, depending on the mass-charge ratio. 

Subextremality in the Reissner--Nordstr\"om family is equivalent to a quasilocal subextremality condition on the photon sphere (see Definition~\ref{subextremal}).

We establish the following uniqueness result for subextremal Reissner--Nordstr\"om spacetimes:

\begin{thm}\label{main} 
 Let $\left(M^n, g, N, \Psi\right)$ be an asymptotically Reissner--Nordstr\"om electrostatic system of mass $m$ and charge $q$, with $n\geq 3$,
 such that $M^n$ has a (possibly disconnected) compact photon sphere as an inner boundary with only subextremal connected components $\Sigma_i^{n-1}$, $1\leq i\leq l$, $l\in\mathbb N$ . 
 Assume moreover that
 \begin{enumerate}
  \item  $N\restriction_{\partial M^n}>0$, and
  \item  $R_{\sigma_i}$ (the scalar curvature of $\Sigma_i^{n-1}$ with respect to the induced metric) is constant.
 \end{enumerate}
 Then $\left(M^n, g\right)$ is isometric to the Reissner--Nordstr\"om manifold of mass $m$ and charge $q$, and $m>|q|$. In particular, $\partial M^n$ is the photon sphere in the Reissner--Nordstr\"om ma\-ni\-fold of mass $m$ and charge $q$, 
 it has only one connected component and is a topological sphere. 
\end{thm}
The above theorem is implied by the more general Theorem~\ref{meta} below. First, we will comment on some of the assumptions of Theorem~\ref{main}:
\begin{rem}

It is natural to require that $N\restriction_{\partial M^n}\geq 0$, meaning that none of the photon sphere component are inside a black hole. 
Moreover, we will see in Proposition~\ref{photonspheregeometry} how the scalar curvature of $\Sigma_i^{n-1}$ relates to the normal derivative of $\Psi$ on $\Sigma_i^{n-1}$, and in view of this relation, 
it is possible to replace the constant scalar curvature condition by the requirement that $|d\Psi|$ be constant on each $\Sigma_i^{n-1}$. 
If $\partial M^n$ has only one connected component, the condition $R_{\sigma_i}=const.$ does not need to be assumed but is fulfilled automatically, as has been argued in~\cite{Israel}, see also the exposition in~\cite{Heusler}: 
in this case, one can show how $\Psi$ can be written as a function of $N$. Since we show in Proposition~\ref{photonspheregeometry} that $|dN|$ is constant on every connected component of the photon sphere, this would imply that $|d\Psi|$ 
and hence $R_{\sigma_i}$ are also constant. However, this is only possible in case the photon sphere has only one connected component. 
Lastly, it is not possible to do away with the quasilocal subextremality condition on the photon sphere; since a superextremal Reissner--Nordstr\"om manifold may contain a photon sphere, 
it may fulfill the assumptions of Theorem~\ref{main} except for the subextremality condition, but the proof of Theorem~\ref{main} must break down in this case due to the absence of a horizon. 

\end{rem}

 The assumptions of Theorem~\ref{main} may be considerably weakened: 
 First, is possible to allow for static horizon components as inner boundary components. Second, one may wish to replace the electrostatic equation for the Ricci tensor (Equation~\eqref{EEVE3}) 
 with the much weaker inequality $N^2R\geq2|d\Psi|^2$. We will refer to objects fulfilling this inequality and the other two electrostatic equations~\eqref{EEVE1} and~\eqref{EEVE2} as \emph{pseudo-electrostatic}, see Section~\ref{electrosection}. 
 
 Since photon spheres in electrostatic settings are characterized by quasilocal geometry, it is useful to define in the pseudo-electrostatic setting the notion of a \emph{quasilocal photon sphere} 
 as a hypersurface characterized by certain quasilocal properties.
 We will see in Propositions~\ref{cmc},~\ref{meanconvex}, and~\ref{photonspheregeometry} that photon spheres in electrostatic systems are always quasilocal photon spheres; hence, the following is a generalization of Theorem~\ref{main}:

\begin{thm}\label{meta}
 Let $\left(M^n, g, N, \Psi\right)$ be a pseu\-do-\-electro\-static system which is asymptotically Reissner--Nordstr\"om of mass $m$ and charge $q$, and $n\geq 3$.
 
Assume $M^n$ has an orientable, compact inner boundary whose connected components are either nondegenerate static horizons or subextremal quasilocal photon spheres. 

Then $\left(M^n, g\right)$ is isometric to a piece of the Reissner--Nordstr\"om manifold of mass $m$ and charge $q$, and $m>|q|$.

 \end{thm}

From Theorem~\ref{meta}, we also immediately get black hole uniqueness in higher-dimensional static asymptotically Reissner--Nordstr\"om electrovacuum:

\begin{cor}
 Let $\left(M^n, g, N, \Psi \right)$ be (pseudo-)electrostatic and asymptotically Reissner--Nordstr\"om (with mass $m$ and charge $q$). 
Assume that $\partial M^n$ is a (possibly disconnected) nondegenerate static horizon. 
 Then $\left(M^n, g \right)$ is isometric to the region of Reissner--Nordstr\"om manifold of mass $m$ and charge $q$ which is outside the horizon, and $m>|q|$.
\end{cor}

This corollary is also a result of~\cite{bubble}, where it was proven under the additional assumption that $m>|q|$. (Note, however, that it is known at least since~\cite{Chrusciel}, where the $3+1$ case was treated, that the mass-charge inequality can be dropped from the assumptions of static electrovacuum uniqueness results.)

The above black hole uniqueness result actually does not require the full electrostatic equations; the pseudo-electrostatic conditions are sufficient 
(in contrast to~\cite{bubble} and most other black hole uniqueness proofs, where the full (electro-)static equations are explicitly required, with the exception of~\cite{carlahigh}).

The proof of Theorem~\ref{meta} uses seminal ideas from the classical black hole uniqueness proofs by Bunting and Masood-ul Alam~\cite{bunting} and by Ruback~\cite{ruback}, 
which were generalized in~\cite{cedergal} and~\cite{cedergalelectro} to prove photon sphere uniqueness results, as well as on the techniques that were developed in~\cite{carlahigh} (see also~\cite{mfo}) to treat higher-dimensional cases.
Its main part breaks down into the following steps: in the first step, we glue to each boundary component an explicitly constructed Riemannian manifold resembling a suitable piece of a Reissner--Nordstr\"om manifold up to a static horizon. 
The horizon allows to reflect the manifold in a second step along its boundary, obtaining an ``upper'' and a ``lower'' half. 
Both the gluing and the doubling can be done with $C^{1,1}$-regularity. 

In a third step, we perform a conformal change of the doubled manifold such that the conformally transformed upper half has vanishing ADM mass, and the conformally transformed lower half can be one-point compactified (with $C^{1,1}$-regularity). 
This will allow to apply a low regularity version of the rigidity case of the positive mass theorem to conclude that the conformally transformed manifold is the Euclidean space. 
In a fourth and last step, uniqueness will be established through some topological arguments and by recovering the conformal factor applying a maximum principle to an elliptic PDE. 

The paper is organized as follows: In Section \ref{setup}, we recall some definitions and known facts about the $n+1$-dimensional Reissner--Nordstr\"om spacetime and about 
asymptotically flat, static electrovacuum spacetimes in general and introduce our asymptotic assumptions. 
In Section~\ref{photonspheres}, we give a definition of photon spheres that is adjusted to our setting and prove some statements about their geometry, which are interesting in their own right. 
Section~\ref{zeromassetc} provides the prerequisites for the conformal change we need to perform in the third step. 
In Section~\ref{masscharge}, we prove the assertions in Theorem~\ref{meta} about mass and charge. Section~\ref{mainproof}  presents the above sketched four steps of the proof of Theorem~\ref{main}.

\section{Setup and definitions}\label{setup}
\subsection{The $n+1$-dimensional Reissner--Nordstr\"om spacetime}

The $n+1$-dimensional Reissner--Nordstr\"om spacetime of mass $m$ and charge $q$ is the manifold $(\mathbb R\times \mathbb R^n\setminus\{0\}, \mathfrak g_{m,q})$, where 
the metric $ \mathfrak g_{m,q}$ is given by
\begin{equation}
 \mathfrak g_{m,q}= -(1-\tfrac{2m}{r^{n-2}}+\tfrac{q^2}{r^{2(n-2)}})dt^2+(1-\tfrac{2m}{r^{n-2}}+\tfrac{q^2}{r^{2(n-2)}})^{-1}dr^2+r^2\Omega_{n-1},\label{RNcoord}
 \end{equation}
and $\Omega_{n-1}$ denotes the standard metric on $\mathbb S^{n-1}$. 
The $n$-dimensional (spatial) Reissner--Nordstr\"om manifold is a canonical spatial slice of the Reissner--Nordstr\"om spacetime, that is, the manifold $\mathbb R^n\setminus\{0\}$ with the metric 

\begin{equation*}
 g_{m,q}=(1-\tfrac{2m}{r^{n-2}}+\tfrac{q^2}{r^{2(n-2)}})^{-1}dr^2+r^2\Omega_{n-1}.
 \end{equation*}

The \emph{lapse} $N_{m,q}$ and the \emph{potential} $\Psi_q$
of the $n$-dimensional Reissner--Nordstr\"om manifold of  mass $m$ and  charge $q$ are the functions 

\begin{align*}
 N_{m,q} (r)&\coloneqq (1-\tfrac{2m}{r^{n-2}}+\tfrac{q^2}{r^{2(n-2)}})^{1/2}, \\
  \Psi_q (r) &\coloneqq \frac{q}{\widehat C r^{n-2}}
\end{align*}

with 
\[
 \widehat C\coloneqq \sqrt{2\frac{n-2}{n-1}}. 
\]

\begin{rem}\label{RNiso}
In isotropic coordinates, $g_{m,q}$ can be written as
\begin{equation}\label{iso}
g_{m,q}=\left(1+\frac{m+q}{2s^{n-2}}\right)^{\frac{2}{n-2}}\left(1+\frac{m-q}{2s^{n-2}}\right)^{\frac{2}{n-2}}\delta\eqdef  \varphi_{m,q}^{\frac{2}{n-2}}\delta,
\end{equation}
where the radial coordinates $s$ and $r$ transform by the rule
\[
 r=s\left(1+\frac{m+q}{2s^{n-2}}\right)^{\frac{1}{n-2}}\left(1+\frac{m-q}{2s^{n-2}}\right)^{\frac{1}{n-2}}.
\]

We can rewrite the lapse and the potential as 
\begin{align*}
 N_{m,q} (s)&=  \frac{\left( 1-\frac{m^2-q^2}{4s^{2(n-2)}}\right) }{\left( 1+\frac{m+q}{2s^{n-2}}\right)\left( 1+\frac{m-q}{2s^{n-2}}\right)},\\
  \Psi_q (s) &= \frac{q}{\widehat C s^{n-2}\left(1+\frac{m+q}{2s^{n-2}}\right)\left(1+\frac{m-q}{2s^{n-2}}\right)}.
\end{align*}

A straightforward computation allows to express $ \varphi_{m,q}$ in terms of $N_{m,q} $ and $\Psi_q $ as
\begin{equation}\label{varphi}
   \varphi_{m,q}=\left(\frac{(N_{m,q} +1)^2-\widehat C^2 \Psi_q^2}{4}\right)^{-1}.
\end{equation}

\end{rem}

In the coordinates of~\eqref{RNcoord}, the outer horizon of the Reissner-Nordstr\"om black hole of  mass $m>0$ and  charge $q$ with $m^2>q^2$
is located at $\left(m+\sqrt{m^2-q^2}\right)^{\frac{1}{n-2}}$. In isotropic coordinates, the location of the outer horizon is at $s_{m, q}\defeq \left(\frac{m^2-q^2}{4}\right)^{\frac{1}{2(n-2)}}$.

\subsection{(Pseudo-)electrostatic spacetimes}\label{electrosection}

The above introduced $n+1$-dimensional Reissner--Nordstr\"om spacetime is the paradigmatic example of a static electrovacuum spacetime. 

Since we will be working in $n$-dimensional spatial slices, it is more convenient to use the dimensionally reduced Einstein--Maxwell equations: 

\begin{defi}
 Let $(M^n, g)$ be a Riemannian manifold and let $N: M^n\rightarrow \mathbb R_{>0}$, $\Psi: M^n\rightarrow \mathbb R$ be smooth functions such that 
\begin{align}
\Delta N &=\frac{\widehat C^2}{N}|d\Psi|^2, \label{EEVE1}\\
0 &=\operatorname{div} \left(\frac{\operatorname{grad}\Psi}{N}\right),\label{EEVE2} \\
N\operatorname{Ric} &=\nabla^2 N-2\frac{d\Psi\otimes d\Psi}{N}+\frac{2}{(n-1)N}|d\Psi|^2g.\label{EEVE3}
\end{align}
 Then $(M^n, g, N, \Psi)$ is called an \emph{electrostatic system}. 
\end{defi}
Here and onwards, $\operatorname{Ric}$ and $R$ denote the Ricci tensor and the scalar curvature of $\left(M^n, g\right)$.

Taking the trace of Equation~\eqref{EEVE3} and plugging in Equation~\eqref{EEVE1}, one obtains
\begin{align}\label{EEVEtrace}
 N^2R=2|d\Psi|^2. 
\end{align}

\begin{defi}\label{almost}
 If $(M^n, g, N, \Psi)$ fulfills Equations~\eqref{EEVE1}, ~\eqref{EEVE2}, and~~\eqref{EEVEtrace}, but not necessarily ~\eqref{EEVE3}, it is called an \emph{traced-electrostatic system}. 
If Equations~\eqref{EEVE1}, ~\eqref{EEVE2} and the inequality 
\begin{equation*}\label{EEVEineq}
N^2R\geq 2|d\Psi|^2
\end{equation*}
are fulfilled, we say that the system is \emph{pseudo-electrostatic}.
\end{defi}

\subsection{Asymptotic considerations}

\begin{rem}[Weighted norms] We will use weighted norms defined as follows: 
\[ \|f\|_{C^2_{-k}(U)}\defeq\sup\limits_{x\in U}||x|^k \cdot |f(x)|+|x|^{k+1} \cdot |D f(x)|+|x|^{k+2} \cdot |D^2f(x)||\] for a twice differentiable function $f$ on an open domain $U\subseteq \mathbb R^n$. 
\end{rem}

We will use the following definition of asymptotically Reissner--Nordstr\"om manifolds: 

\begin{defi}\label{asyiso}
A smooth Riemannian manifold $(M^n, g)$ of dimension $n\geq 3$ is called \emph{asymptotically Reissner-Nordstr\"om} of  mass $m$ and  charge $q$ if 
\begin{enumerate}
\item $M^n$ is diffeomorphic to $K\sqcup E$, where $K$ is a compact set, and $E$ is an \emph{asymptotic end} which is diffeomorphic to $\mathbb R^n\setminus \overline{B^n_S(0)}$ for some $S>s_{m,q}$, 

\item for the diffeomorphism $\Phi=(x^i): E^n\rightarrow \mathbb R^n\setminus \overline{B^n_S(0)}$ and the metric $g$ there is a constant $C>0$ such that 
\[\|(\Phi_\ast g)_{ij}-(g_{m,q})_{ij}\|_{C^2_{-(n-1)}(\mathbb R^n\setminus \overline{B^n_S(0)})} \leq C, \ i,j=1, \dots , n,
\]
\item $\Phi_\ast g$ is uniformly positive definite and uniformly continuous on 
$\mathbb R^n\setminus \overline{B^n_S(0)}$. 
\end{enumerate}
Here, $(g_{m,q})_{ij}$ are the components of the Reissner-Nordstr\"om metric in isotropic coordinates, see Equation~\eqref{iso}. 
\end{defi}

\begin{nota} We will often notationally omit $\Phi$ and $\Phi_\ast$ whenever this does not lead to ambiguity. Moreover, we use the coordinates $\left(x^i\right)$ defined by $\Phi$ 
to define a radial coordinate $s\defeq \sqrt{\sum\limits_{i=1}^n|x^i|^2}$ on the asymptotic end. 
\end{nota}

\begin{defi}\label{asyrnfct}
Let $(M^n, g)$  be an asymptotically Reissner--Nordstr\"om manifold of  mass $m$ and  charge $q$ with an asymptotic end $E^n$, and $\Phi=(x^i): E^n\rightarrow \mathbb R^n\setminus \overline{B^n_S(0)}$
a diffeomeorphism as in Definition~\ref{asyiso}. 
A smooth function $N: M^n\rightarrow \mathbb R$ is called an \emph{asymptotic Reissner--Nordstr\"om lapse (of  mass $m$)} if there is a constant $C$ such that
\[
\| \Phi_\ast N- N_{m,q}\|_{C^2_{-(n-1)}(\mathbb R^n\setminus \overline{B^n_S(0)})} \leq C
\]
for some (hence, all) $q\in\mathbb R$.  

A smooth function $\Psi: M^n\rightarrow \mathbb R$ is called an \emph{asymptotic Reissner--Nordstr\"om potential (of  charge $q$)} if there is a constant $C$ such that
\[
\| \Phi_\ast \Psi- \Psi_q\|_{C^2_{-(n-1)}(\mathbb R^n\setminus \overline{B^n_S(0)})} \leq C.
\]

A quadruple $(M^n, g, N, \Psi)$ is called an \emph{asymptotically Reissner--Nordstr\"om system} if $(M^n, g)$ an asymptotically Reissner--Nordstr\"om manifold of 
mass $m$ and  charge $q$, $N: M^n\rightarrow\mathbb \R$ is an asymptotic Reissner--Nordstr\"om lapse of the same mass $m$, and $\Psi: M^n\rightarrow \mathbb R$ is an asymptotic Reissner--Nordstr\"om potential of the same charge $q$. 

\end{defi}

\section{Quasilocal geometry}\label{photonspheres}
We cite the following fundamental theorem:
	
\begin{thm}[\cite{CVE, perlick}]
  A timelike hypersurface $P$ in a Lorentzian manifold is totally umbilical if and only if every lightlike geodesic that is initially tangent to $P$ stays tangent to $P$ for as long as it exists. 
\end{thm}

We remind the reader that a submanifold is called totally umbilical if the trace-free part of its second fundamental form vanishes.
 
The above theorem motivates the following definition, following~\cite{CVE, yazalazov, cederphoto}: 
\begin{defi}\label{photonspheredefi}
Let $(M^n, g, N, \Psi)$ be a pseudo-electrostatic system (see Definition~\ref{almost}). 

  A timelike embedded orientable hypersurface $P^n$ in $ (\mathbb R\times M^n, -N^2dt^2+g)$ is called a \emph{photon sphere} 
  if it is totally umbilical 
	and $N$ and $\Psi$ are constant on every connected component of $P^n$. 
  \end{defi}
It is immediate that $P^n\cap M^n$ is totally umbilical in $M^n$. If $P^n$ is a photon sphere, we will occasionally also refer to  $P^n\cap M^n$ as a photon sphere. 
 
In the $n+1$-dimensional Reissner--Nordstr\"om spacetime with mass parameter $m>0$, there is a photon sphere located at the radius 
\[ \left(\frac 1 2 mn+\frac 1 2\sqrt{m^2n^2-4(n-1)q^2}\right)^{\tfrac{1}{n-2}}\] (in the coordinates of~\eqref{RNcoord}), provided that the mass-charge ratio is such that $m^2n^2-4(n-1)q^2$ is nonnegative (which is always the case in subextremal Reissner--Nordstr\"om spacetime), while a Reissner--Nordstr\"om spacetime of negative or vanishing mass does not possess a photon sphere.

For the rest of this paper, we fix the following notation: 
\begin{nota}
 For a photon sphere $(P^n, p)\hookrightarrow (\mathbb R\times M^n, -N^2dt^2+g)$ of an pseudo-electrostatic system $(M^n, g, N, \Psi)$, we write 
\[
(P^n,p)=
=\bigcup\limits_{i=1}^l (\mathbb R\times \Sigma_i^{n-1}, -N_i^2 dt^2+\sigma_i),
\]
where each $\mathbb R\times \Sigma_i^{n-1}$ is a connected component of $P^n$. 

We define 
\begin{align*}
 N_i\defeq N\restriction_{\Sigma_i^{n-1}},\\
 \Psi_i\defeq \Psi\restriction_{\Sigma_i^{n-1}}. 
\end{align*}

Moreover, $\mathfrak H$ denotes the mean curvature of $\bigcup\limits_{i=1}^l\left(\mathbb R\times \Sigma_i^{n-1}\right)$, while $H$ denotes the mean curvature of $\bigcup\limits_{i=1}^l \Sigma_i^{n-1}$ in $M^n$, and we set

\begin{align*}
 \mathfrak H_i&\defeq \mathfrak H\restriction_{\mathbb R\times \Sigma_i^{n-1}},\\
 H_i&\defeq  H\restriction_{ \Sigma_i^{n-1}}.
\end{align*}

A choice of unit normal to $\bigcup\limits_{i=1}^l\Sigma_i^{n-1}$ in $M^n$ (pointing towards the asymptotic end if $(M^n, g)$ is asymptotically flat) will be denoted by $\nu$, and we set
\begin{align*}
 \nu (N)_i &\defeq \nu(N)\restriction_{\Sigma_i^{n-1}}\\
 \nu (\Psi)_i&\defeq \nu(\Psi)\restriction_{\Sigma_i^{n-1}}. 
\end{align*}
\end{nota}

Photon spheres in the electrostatic setting are characterized by quasilocal properties which make them a \emph{quasilocal photon sphere} defined as follows: 

\begin{defi}\label{quasiphoton}
 A totally umbilical hypersurface in a (pseudo-)electrostatic system $\left(M^n, g, N, \Psi\right)$ that fulfills $\nu(N)_i>0$, $N_i=const.$, $R_{\sigma_i}=const.>0$, $H_i=const.$ is called a \emph{quasilocal photon sphere component}
if the equations
\begin{equation}\label{scalar2}
R_{\sigma_i}=\frac{n}{n-1}H_i^2+\frac{2}{N_i^2}\nu(\Psi)_i^2. 
\end{equation}
and
\begin{equation}\label{NandH}
     \frac{H_i}{N_i}\nu(N)_i=\frac{H_i^2}{n-1}
    \end{equation}
are fulfilled. 
\end{defi}

The following three propositions serve to show that a photon sphere component in an electrostatic system is a quasilocal photon sphere component. 
While Propositions~\ref{cmc} and~\ref{meanconvex} are straightforward generalizations from the $3+1$-dimensional setting (see~\cite{yazalazov} and~\cite{cederphoto} for Proposition~\ref{cmc} and ~\cite{cedergal} and~\cite{cedergalelectro} for Proposition~\ref{meanconvex}), we will prove Proposition~\ref{photonspheregeometry} for our setting. If $\Psi=0$, the equations reduce to the ones given for the curvature quantities of photon spheres in~\cite{carlahigh}. 
For a similar proof in dimension $3$, see also~\cite{yazalazov}.

\begin{prop}[\cite{yazalazov}, \cite{cederphoto}]\label{cmc}
Let $(M^n, g, N, \Psi)$ be an electrostatic system and $P^n$ a photon sphere in $ (\mathbb R\times M^n, \mathfrak g\defeq -N^2dt^2+g)$ with induced metric $p$.
Then for every $1\leq i\leq l$, $\mathfrak H_i$ and $H_i$ are constant. 
\end{prop}

\begin{prop}[\cite{galmiao, cedergal, cedergalelectro}]\label{meanconvex}
Let $P^n$ be a quasilocal photon sphere in $ (\mathbb R\times M^n, \mathfrak g\defeq -N^2dt^2+g)$, and assume that  $ (\mathbb R\times M^n, \mathfrak g)$ fulfills the null energy condition. Then $H_i>0$. 
\end{prop}

\begin{prop}\label{photonspheregeometry}
Let $(M^n, g, N, \Psi)$ be an electrostatic system and let $(P^n,p)$ be a photon sphere in $(\mathbb R\times M^n, -N^2dt^2+g)$. 

Then Equation \eqref{scalar2} holds. In particular, $R_{\sigma_i}$ is nonnegative, and it is positive provided that $H_i\neq 0$.

Moreover, Equation~\eqref{NandH} holds, 
 $H_i$ and  $\nu (N)_i$ are constant; and $R_{\sigma_i}$  is constant if and only if $\nu(\Psi)_i$ is. 
\end{prop}

\begin{proof}

First we show Formula~\eqref{scalar2}. 
We write $\mathfrak{Ric}$ and $\mathfrak{R}$ for the Ricci tensor and the scalar curvature of $(\mathbb R\times M^n, -N^2dt^2+g)$ and choose $\eta\defeq\frac 1 N \partial_t$ as a unit normal to 
$\bigcup\limits_{i=1}^l\Sigma_i^{n-1}$ in $P^n$. One calculates (applying a general formula for the curvature of warped products to $\left(\mathbb R\times M^n, -N^2dt^2+g\right)$, see e.g.~\cite{ONeill1}, and using Equation~\eqref{EEVE1}), that
\begin{equation*}
 \mathfrak{Ric}(\eta, \eta)=\frac{\Delta N}{N}=\frac{\widehat C^2}{N^2}|d\Psi|^2. 
\end{equation*}

The traced Gauss equation and the fact that $M^n$ is totally geodesic in $ \mathbb R\times M^n$ give 
\[
 \mathfrak R= R-2 \mathfrak{Ric} (\eta, \eta), 
\]
so that we arrive at 
\begin{align}
 \mathfrak R &= \frac{2}{N^2}|d\Psi|^2\left(1-\widehat C^2\right) \\
 &= \frac{2 \nu(\Psi)^2}{N^2}\left(1-\widehat C^2\right),\label{temp1.1}
\end{align}
where we also used Equation~\eqref{EEVEtrace} and the fact that  $ \Psi$ is constant on $P^n$ by definition of photon spheres. 

The traced Gauss equation applied to $P^n\hookrightarrow \mathbb R\times M^n$ simplifies by umbilicity to 

\begin{equation}\label{temp1.2}
 \mathfrak R-2 \mathfrak{Ric}(\nu, \nu) = R_P-\frac{n-1}{n}\mathfrak H^2,
\end{equation}

where $R_P$ denotes the scalar curvature of $P^n$, and we abused notation by denoting the unit normal to $P^n$ in $\mathbb R\times M^n$ by $\nu$ (like the unit normal to  $\bigcup\limits_{i=1}^l\Sigma_i^{n-1}$ in $M^n$). 

Again by a standard warped product formula,
\[
 \mathfrak{Ric}(\nu, \nu)=\operatorname{Ric}(\nu, \nu)-\frac{1}{N}\nabla^2 N (\nu, \nu), 
\]
so that by Equation~\eqref{EEVE3} and the fact that $ \Psi$ is constant on $P^n$, 

\begin{equation}\label{temp1.3}
 \mathfrak{Ric}(\nu, \nu)=-\widehat C^2 \frac{\nu (\Psi)^2}{N^2}. 
\end{equation}

Combining Equations~\eqref{temp1.1},~\eqref{temp1.2}, and~\eqref{temp1.3} allows to express $R_P$ as 

\begin{align*}
 R_P&= \frac{n-1}{n}\mathfrak H^2 +\frac{2\nu(\Psi)^2}{N^2}\left(1-\widehat C^2+\widehat C^2\right)\\
 &=\frac{n}{n-1} H^2 +\frac{2\nu(\Psi)^2}{N^2},
\end{align*}
and we have shown Equation\eqref{scalar2}.

We now prove Formula~\eqref{NandH}. 

By the traced Gauss equation and by umbilicity of $\Sigma_i^{n-1}$ in $M^n$, 
\begin{equation}\label{gaussforscalar}
\operatorname{R}-2\operatorname{Ric}(\nu, \nu)=R_{\sigma_i}-\frac{n-2}{n-1}H_i^2. 
\end{equation}
Plugging $\nu$ into both slots of Equation~\eqref{EEVE3} and the fact that $\Psi$ is constant on $\Sigma_i^{n-1}$ give
\begin{equation}
\operatorname{Ric}(\nu, \nu)=\frac{\nabla^2 N}{N}(\nu, \nu)-2\frac{n-2}{n-1}\frac{\nu(\Psi)^2}{N^2}.
\end{equation}

Recall that in general for a smooth isometric embedding of manifolds $(M_1^{n-1}, g_1)\hookrightarrow (M_2^n, g_2)$ with a spacelike unit normal $\nu$ and a smooth function $f: M_2^n\rightarrow \mathbb R$,
the formula 
\begin{equation}\label{decompolaplacian}
{}^{g_2}\Delta f={}^{g_1}\Delta f +{}^{g_2}\nabla^2(\nu, \nu) +H_{M_1}\nu(f)
\end{equation}
holds, where $H_{M_1}$ denotes the mean curvature of $M_1$ in $M_2$. 

Using this and Equation~\eqref{EEVE1}, we get 
\begin{align*}
\nabla^2 N (\nu, \nu)&=\Delta N-\Delta_{\sigma_i} N-H_i \nu(N)\\
&=\Delta N-H_i \nu(N)\\
&=2\frac{n-2}{n-1}\frac{\nu(\Psi)^2}{N}-H_i \nu(N). 
\end{align*}
This gives 

\begin{equation}
\operatorname{Ric}(\nu, \nu)=-\frac{H_i\nu(N)}{N}.
\end{equation}
Likewise, Equation~\eqref{EEVEtrace} reads in our case
\begin{equation}
R=2\frac{\nu(\Psi)^2}{N^2}.
\end{equation}
Plugging these expressions for $\operatorname{Ric}(\nu, \nu)$ and $\operatorname{R}$ into Equation~\eqref{gaussforscalar}, we get
\begin{equation}\label{scalar1}
2\frac{\nu(\Psi)_i^2}{N_i^2}+2\frac{H_i\nu(N)_i}{N_i}=R_{\sigma_i}-\frac{n-2}{n-1}H_i^2.
\end{equation}

Equation~\eqref{NandH} now follows immediately from Equations~\eqref{scalar2} and~\eqref{scalar1}. 

We note that
\[
 nH_i=(n-1)\mathfrak H_i,
\]
and since $\mathfrak H_i$ is constant (by Proposition~\ref{cmc}), so is $H_i$. 
The assertions about constancy of $\nu(N)_i$, $\nu(\Psi)_i$, and $ R_{\sigma_i}$ follow directly from Equations~\eqref{scalar2} and~\eqref{NandH}.

\end{proof}

Moreover, we define a notion of subextremality for photon spheres and quasilocal photon spheres, in agreement with the definition in~\cite{cedergalelectro}: 

\begin{defi}\label{subextremal}
A (quasilocal) photon sphere component $\Sigma_i^{n-1}$ in a (pseudo-)electrostatic system is called \emph{subextremal} if 
\[
 \frac{H_i^2}{R_{\sigma_i}}>\frac{n-2}{n-1}.
\]
If  ``$<$'' (``$=$'') holds, it is called \emph{superextremal (extremal)}. 
\end{defi}

As usual, a hypersurface $\Sigma^{n-1}$ in $M^n$ is called a \emph{static horizon} if it is the zero level set of the static lapse $N$, and it is called \emph{nondegenerate} if the outer normal derivative of the lapse is positive. 

We recall some well-known facts about static horizons in the electrostatic setting that carry over to the pseudo-electrostatic case: 

\begin{lem}\label{horizon} Let $\left(M^n, g, N, \Psi\right)$ be a (pseudo-)electrostatic system and  $\Sigma^{n-1}\subseteq M^n$ a static horizon. Then 
\begin{enumerate}
\item $\Sigma^{n-1}$ has vanishing mean curvature,
\item $\Psi\restriction_{\Sigma^{n-1}}= const.$, 
\item $d\Psi\restriction_{\Sigma^{n-1}}=0$. 
\end{enumerate}
\end{lem} 
For a proof, we refer the reader to the derivation of Equations~(11) and~(13) in~\cite{bubble}, where these statements were deduced in an electrostatic context, but without appealing to the electrostatic equation for the Ricci tensor~\eqref{EEVE3}. 

\section{Zero mass and one-point insertion}\label{zeromassetc}

In this section, we prove two propositions about the asymptotic behavior of Reissner--Nordstr\"om manifolds after a specific conformal change which will be used in the proof of the main results Theorem~\ref{main} and Theorem~\ref{meta}. 

\begin{prop}[Zero mass of an asymptotic end after a conformal change]\label{zero}

\hfill\break
 Let $(M^n, g, N, \Psi)$ an asymptotically Reissner--Nordstr\"om system of  mass $m$ and  charge $q$. 

 Assume that $\Omega_{+} \defeq \left( \frac {(1+ N)^2-\widehat C^2\Psi^2} 4 \right)^{1/(n-2)}>0$ on all of $M^n$. 
 
Then the metric $ \Omega_+ ^2 g$ is asymptotically Reissner--Nordstr\"om with mass~$0$ and charge~$0$. 
\end{prop}
\begin{proof}
We write $\Phi: E\rightarrow \mathbb R^n$ for the diffeomorphism that makes $M^n$ asymptotically Reissner--Nordstr\"om as in Definition~\ref{asyiso} and recall that we required $S>s_{m,q}$ (see Definition~\ref{asyiso}).

Since  $N_{m,q}$ and $\Psi_q$ are given explicitly, we may check that 
\begin{align}\label{phiisgood}
\left\|\left( \varphi_{m,q}\right)^{-\tfrac{2}{n-2}}- \left(\tfrac{(N_{m,q} +1)^2-\widehat C^2 \Psi_q^2}{4}\right)^{\tfrac{2}{n-2}}\right\|_{C^2_{0}(\mathbb R^n\setminus \overline{B^n_S(0)})}<C_1
\end{align}

(see Equation \eqref{varphi}) on $\mathbb R^n\setminus \overline{B^n_S(0)}$ for some $C_1=C_1(n, S)$.

The asymptotic behavior  
\[\| \Phi_\ast N- N_{m,q}\|_{C^2_{-(n-1)}(\mathbb R^n\setminus \overline{B^n_S(0)})} \leq C_2\] and \[\| \Phi_\ast \Psi- \Psi_q\|_{C^2_{-(n-1)}(\mathbb R^n\setminus \overline{B^n_S(0)})} \leq C_2\]
 gives that 

\begin{equation}\label{omegaphi}
  \|(\Phi_\ast \Omega_+^2) - \varphi_{m,q}^{-\frac{2}{n-2}} \|_{C^2_{-(n-1)}(\mathbb R^n\setminus \overline{B^n_S(0)})}<C_3
\end{equation}
for some $C_3=C_3(C_2,n, S)$.

From these facts combined we conclude that 

\[
 \left\|\Phi_\ast\Omega_+^2\right\|_{C^2_{0}(\mathbb R^n\setminus \overline{B^n_S(0)})}<C_4
\]

for some $C_4=C_4(C_1, C_3, n, S)$. 

Using the assumption that $\|(\Phi_\ast g)_{ij}-(g_{m,q})_{ij}\|_{C^2_{-(n-1)}(\mathbb R^n\setminus \overline{B^n_S(0)})} \leq C_5 $
for some $C_5$ and all $ i,j=1, \dots , n$, we now get 
\begin{equation*}
\|(\Phi_\ast \Omega_+^2) (\Phi_\ast g)_{ij}-(\Phi_\ast \Omega_+^2) (g_{m,q})_{ij}\|_{C^2_{-(n-1)}(\mathbb R^n\setminus \overline{B^n_S(0)})}<C_6
\end{equation*}
for some $C_6=C_6(C_4, C_5, n, S)$ and all $ i,j = 1, \dots , n$. 

On the other hand, the inequalities \eqref{phiisgood} and \eqref{omegaphi} also imply that there is a $C_7=C_7(C_1, C_3, n, S)$ such that
\begin{align*}
&\|(\Phi_\ast \Omega_+^2) (g_{m,q})_{ij}-  \varphi_{m,q}^{-\frac{2}{n-2}} (g_{m,q})_{ij}\|_{C^2_{-(n-1)}(\mathbb R^n\setminus \overline{B^n_S(0)})}\\
& = 
\| \varphi_{m,q}^{-\frac{2}{n-2}}\left(\Phi_\ast \Omega_+^2\delta_{ij} - (g_{m,q})_{ij}\right)\|_{C^2_{-(n-1)}(\mathbb R^n\setminus \overline{B^n_S(0)})}<C_7
\end{align*}
for all $i,j=1, \dots , n$. 

We can now compute
\begin{align*}
 &\|(\Phi_\ast \Omega_+^2 g)_{ij}-(g_{m=0,q=0})_{ij}\|_{C^2_{-(n-1)}(\mathbb R^n\setminus \overline{B^n_S(0)})} \\
 =&\|(\Phi_\ast \Omega_+^2 g)_{ij}-\delta_{ij}\|_{C^2_{-(n-1)}(\mathbb R^n\setminus \overline{B^n_S(0)})} \\
  =&\|(\Phi_\ast \Omega_+^2) (\Phi_\ast g)_{ij}- \varphi_{m,q}^{-\frac{2}{n-2}} (g_{m,q})_{ij}\|_{C^2_{-(n-1)}(\mathbb R^n\setminus \overline{B^n_S(0)})} \\
\leq &  \|(\Phi_\ast \Omega_+^2) (\Phi_\ast g)_{ij}-(\Phi_\ast \Omega_+^2) (g_{m,q})_{ij}\|_{C^2_{-(n-1)}(\mathbb R^n\setminus \overline{B^n_S(0)})} \\
 & + \|(\Phi_\ast \Omega_+^2) (g_{m,q})_{ij}- \varphi_{m,q}^{-\frac{2}{n-2}} (g_{m,q})_{ij}\|_{C^2_{-(n-1)}(\mathbb R^n\setminus \overline{B^n_S(0)})}\\
  < &C_6+C_7,
\end{align*}
which proves that $(M^n, \Omega_+ ^2 g)$ is asymptotically Reissner-Nordstr\"om with  mass~$0$ and  charge~$0$.
\end{proof}

\begin{prop}[One-point insertion]\label{onepoint}
 Let $(M^n, g, N, \Psi)$ an asymptotically Reissner--Nordstr\"om system of  mass $m$ and  charge $q$. 

 Assume that $\Omega_{-} \defeq \left( \frac {(1-N)^2-\widehat C^2\Psi^2} 4 \right)^{1/(n-2)}>0$ on all of $M^n$. 
 
Then one can insert a point $p_\infty$ into $(M^n, \Omega_{-} ^2 g)$ to obtain a compact Riemannian manifold $(M^n_\infty \coloneqq M^n\cup \{p_\infty\}, g_\infty)$ which is $C^{1,1}$-regular at $p_\infty$ and with boundary $\partial M^n_\infty =\partial M^n$.

\end{prop}

\begin{proof}
Again, we write $\Phi: E\rightarrow \mathbb R^n$ for the diffeomorphism that makes $M^n$ asymptotically Reissner--Nordstr\"om as in Definition~\ref{asyiso} and recall that we required $S>s_{m,q}$. 

We note that 
\[
(1- N_{m, q})^2-\widehat{C^2}{\Psi}^2_{q}=\frac{m^2-q^2}{s^{2(n-2)}\left(1+\frac{m+q}{2s^{n-2}}\right)\left(1+\frac{m-q}{2s^{n-2}}\right)} 
\]

and hence 
\[
\frac{(1- N_{m, q})^2-\widehat{C^2}{\Psi}^2_{q}}{4} \cdot  \varphi_{m,q}=\frac{m^2-q^2}{4s^{2(n-2)}}=\left(\frac{s_{m,q}}{s}\right)^{2(n-2)}. 
\]

With similar arguments as in the proof of Proposition~\ref{zero} and following closely~\cite{carlahigh}, we can use this and the asymptotic behavior of $N_{m,q}$ and $\Psi_q$ to estimate 
\begin{align*}
 &\left\|\Phi_\ast (\Omega_{-} ^2 g)_{ij}-\left(\tfrac{s_{m,q}}{s}\right)^4\delta_{ij} \right\|_{C^2_{-(n+3)}(\mathbb R^n\setminus \overline{B^n_S(0)})}\\
 =&\left\|\Phi_\ast \left(\left( \tfrac {(1-N)^2-\widehat C^2\Psi^2} 4 \right)^{2/(n-2)} g_{ij}\right)-\left(\tfrac{(1-N_{m,q})^2-
 \widehat C^2\Psi_q^2}{4}\right)^{\tfrac{2}{n-2}}(g_{m,q})_{ij} \right\|_{C^2_{-(n+3)}(\mathbb R^n\setminus \overline{B^n_S(0)})}\\
 \leq & C_1
\end{align*}
for some $C_1=C_1(m, q, n, S)$ and all $i,j=1, \dots , n$, where we used additionally that 
\[
\left\|\left(\tfrac{(1-N)^2-\widehat C^2\Psi^2}{4}\right)^{\tfrac{2}{n-2}}\right\|_{C^2_{-4}(\mathbb R^n\setminus \overline{B^n_S(0)})}\leq C_0
\]
for some $C_0=C_0(m, q, n, S)$. 

Analogously to the proof of Proposition 2.8 in~\cite{carlahigh}, let now $(y^i)$ denote coordinates in $\R^n\setminus \overline{B^n_S(0)}$ so that $s=\vert y\vert_\delta$, 
and perform an inversion at the sphere of radius $s_{m,q}$ to new coordinates $\eta^{i}\coloneqq \left(\frac{s_{m,q}}{s}\right)^2 y^{i}$. Then
\begin{align*}
(\Phi_{\ast}\Omega_{-} ^2 g\,)(\partial_{\eta^{k}},\partial_{\eta^{l}})
&=(\Phi_\ast\Omega_{-} ^2 g\,)(\partial_{y^{i}},\partial_{y^{j}})\left(\frac{s}{s_{m,q}}\right)^{4}\left(\delta^{i}_{k}-2\frac{y^{i}y_{k}}{s^{2}}\right)\left(\delta^{j}_{l}-2\frac{y^{j}y_{l}}{s^{2}}\right),\\[1ex]
\delta(\partial_{\eta^{k}},\partial_{\eta^{l}})&
=\delta(\partial_{y^{i}},\partial_{y^{j}})\left(\frac{s}{s_{m,q}}\right)^{4}\left(\delta^{i}_{k}-2\frac{y^{i}y_{k}}{s^{2}}\right)\left(\delta^{j}_{l}-2\frac{y^{j}y_{l}}{s^{2}}\right)
=\left(\frac{s}{s_{m,q}}\right)^{4}\delta_{kl},
\end{align*}
where the indices are lowered and raised with the flat metric $\delta$. 

Together with the above estimate, it follows that
\begin{align*}
&\left\Vert 
(\Phi_{\ast}\Omega_-^2 g\,)(\partial_{\eta^{k}},\partial_{\eta^{l}})-\delta_{kl}\right\Vert_{C^{2}_{-(n-1)}(\R^n\setminus \overline{B^n_S(0)})}\\
= & \left\Vert
\left[\Phi_\ast (\Omega_{-} ^2 g)_{ij}-\left(\tfrac{s_{m,q}}{s}\right)^4\delta_{ij} \right]
 \left(\tfrac{s}{s_{m,q}}\right)^4 
\left(\delta^i_k-2\tfrac{y^iy_k}{s^2}\right)
\left(\delta^j_l-2\tfrac{y^j y_l}{s^2}\right) 
\right\Vert_{C^{2}_{-(n-1)}(\R^n\setminus \overline{B^n_S(0)})}\\
\leq & C_2
\end{align*}
for some $C_2=C_2(C_0, C_1, m, q, n, S)$, where the subscript ${C^{2}_{-(n-1)}(\R^n\setminus \overline{B^n_S(0)})}$ is to be interpreted in $(y^i)$-coordinates. 

In terms of the new coordinates $(\eta^i)$ and writing $S' \coloneqq \frac{s_{m,q}^2}{S}$, this (along with the assumption that $n\geq 3$) allows to conclude that 
\begin{align*}
&\left\Vert 
(\Phi_{\ast}\Omega_-^2 g\,)(\partial_{\eta^{k}},\partial_{\eta^{l}})-\delta_{kl}\right\Vert_{C^{2}_{2}(B^n_{S'}(0))}<C_3
\end{align*}
for some $C_3=C_3(C_0, C_1, C_2, m, q, n, S)$. 
We can thus insert a point $p_\infty$ (with $\eta^i(p_\infty)=0$ for all $i=1, \dots, n$) into $M^n$ and extend $g$ to a metric $g_\infty$ on $M^n_\infty=M^n\cup \{p_\infty\}$ by letting
\[
 g_\infty (x)\defeq
\left\{\begin{array}{lr} \Omega_-^2g(x) &\text{ for } x\neq p_\infty,\\
\delta &\text{for } x=p_\infty,
        \end{array}
\right. \]
and $g_\infty$ has $C^{1,1}$-regularity at $p_\infty$. 
\end{proof}

\section{Mass and charge}\label{masscharge}

We now prove two lemmata which together show $m>|q|$ under assumptions that are weaker than those of Theorem~\ref{meta}.  

\begin{lem}\label{aftermain1}
Let $\left(M^n, g, N, \Psi\right)$ be asymptotically Reissner--Nordstr\"om and let Equation~\eqref{EEVE1} be fulfilled. 
If $N$ is constant on each component of $\partial M^n$ and $\nu(N)>0$ on the inner boundary $\partial M^n$, then $m>0$. 
\end{lem}
\begin{proof}
Let us first assume that $N>0$ on all of $\partial M^n$. 
By Stokes' theorem and Equation~\eqref{EEVE1}, 
\begin{align*}
 0&<\int\limits_{\partial M^n}^{}\nu(N)=-\int\limits_{M^n}^{} \Delta N+\int\limits_{\mathbb S_\infty^{n-1}}^{} \nu(N)\\
 &=-\int\limits_{M^n}^{} \frac{\widehat C^2}{N}|d\Psi|^2+\int\limits_{\mathbb S_\infty^{n-1}}^{} \nu(N)\\
 & \leq \int\limits_{\mathbb S_\infty^{n-1}}^{} \nu(N),
\end{align*}

where $\mathbb S_\infty^{n-1}$ is a sphere at infinity. By the asymptotic behavior of $N$,
\[
\int\limits_{\mathbb S_\infty^{n-1}}^{} \nu(N)=(n-2)\operatorname{vol}\left(\mathbb S_1^{n-1}\right)m
\]
so $m$ is positive. 

If $\partial M^n$ has components where $N$ vanishes, we pass from these components to a close-by level surface of $N$ where $N>0$ and $\nu(N)>0$. 
\end{proof}

For the remainder of this article, $\left(M^n, g, N, \Psi \right)$ will be as in the assumptios of Theorem~\ref{meta}. 
We will continue to denote those boundary components of $\left(M^n, g, N, \Psi \right)$ which are quasilocal photon spheres by $\Sigma_i^{n-1}$ ($1\leq i\leq l$) (in agreement with the notation fixed in Section~\ref{photonspheres}),
while the horizon components will be denoted by $\widehat\Sigma_i^{n-1}$ ($l+1\leq i\leq L$).

\begin{lem}\label{mandq}
For $(M^n, g, N, \Psi)$ consider the following conditions: 

\begin{enumerate}
\item each quasilocal photon sphere component is subextremal and each static horizon component is nondegenerate, 
\item $ F_\pm\defeq N-1\pm\widehat C  \Psi<0$ on $M^n$, 
\item $m^2>q^2$.
\end{enumerate}
\end{lem}
Then $(1)\Rightarrow (2)\Rightarrow (3)$.
\begin{proof}
``$(1)\Rightarrow (2)$'':
Due to the electrostatic equations~\eqref{EEVE1} and~\eqref{EEVE2}, $F_\pm$ fulfills 
\begin{equation}\label{pde}
 \Delta F_\pm \mp\frac{\widehat C d\Psi (\operatorname{grad}F_\pm)}{N}=0,
\end{equation}
 
 which we will use to apply a maximum principle to $F_\pm$. 
On the other hand, the asymptotic behavior of $ N$ and $ \Psi$ gives $F_\pm\rightarrow 0$ as $r\to\infty$. 
If $F_\pm $ was positive at some point in $M^n$, then $F_\pm $ had a positive maximum on $M^n\cup \partial M^n$, hence (by the maximum principle) on $\partial M^n$, which means that at least on one boundary component $\Sigma_i^{n-1}$, the normal derivative $\nu(F_\pm)$ could not be positive.
Hence, $F_\pm$ is negative on $M^n$ provided that $\nu(F_\pm)$ is positive on  $\partial M^n$; it thus remains to be shown that  $\nu(F_\pm)\restriction_{\partial M^n}>0$:

We fix $1\leq i\leq l$. 
If $\Sigma_i^{n-1}$ is a quasilocal photon sphere component, then 
$\nu(F_\pm)$ is positive on $\Sigma_i^{n-1}$ if and only if 
\begin{align*}
 & \nu(N)^2-\widehat C^2\nu(\Psi)^2>0\\
 \xLeftrightarrow{\mathit{\eqref{scalar2}}, \eqref{NandH}}  &  \frac{H_i^2}{(n-1)^2}-\frac{n-2}{n-1}\left(  R_{\sigma_i}-\frac{n}{n-1}H_i^2\right)>0\\
\Leftrightarrow & \frac{H_i^2}{R_{\sigma_i}}>\frac{n-2}{n-1},
\end{align*}

and the last inequality is the subextremality condition. 

For static horizon components, $\nu(F_\pm)\restriction_{\Sigma_i^{n-1}}>0$ was shown in \cite{bubble} as a consequence of Equation~\eqref{EEVE1} and the nondegeneracy condition $\nu(N)>0$, by analyzing the near-horizon asymptotics.

We have shown that $F_\pm$ is negative on $M^n$ provided that every 
quasilocal photon sphere component is subextremal and every static horizon component is nondegenerate. 

``$(2)\Rightarrow (3)$'':

For the Reissner--Nordstr\"om lapse and potential, we note the asymptotic behavior 
\[
N_{m,q}-1\pm \widehat C\Psi_q=(-m\pm q)r^{-n+2}+\mathcal O (r^{-n+1})
\]

for $r\to\infty$. 

By the asymptotic conditions for $N$ and $\Psi$, we deduce that also 
\[
F_\pm =(-m\pm q)r^{-n+2}+\mathcal O (r^{-n+1}) 
\]

for $r\to\infty$.

By the assumption, $F_\pm$ is negative in the asymptotic region, and hence $m>\pm q$. 
\end{proof}

Note that the implication ``$(2)\Rightarrow (3)$'' was already shown in~\cite{cedergalelectro} for the case that the boundary is a photon sphere and $n=3$; our proof here is very similar.

\section{Main part of the proof of the main results}\label{mainproof}

In the remainder of this article, we will prove Theorem~\ref{meta}. Since we have shown in Section~\ref{photonspheres} that a photon sphere in an electrostatic system is a quasilocal photon sphere, this will imply Theorem~\ref{main}.  

The proof follows the steps of constructing suitable pseudo-electrostatic fill-ins and attaching them to the boundary, doubling this new manifold along its new boundary, conformally compactifying it and applying the positive mass theorem, and determining the conformal factor to show that the original manifold was a piece of a Reissner--Nordstr\"om manifold. 

\subsection{Gluing in pseudo-Reissner--Nordstr\"om necks}\label{glue}

In this section, we glue suitable pieces of pseudo-electrostatic spacetimes  to quasilocal photon sphere components of the inner boundary of the (pseudo-)electrostatic system $(M^n, g, N, \Psi)$, thereby getting a horizon as a new inner boundary. 
 
 We fix a quasilocal photon sphere $\Sigma_i^{n-1}$ and define its scalar curvature radius
 \[
  r_i\defeq \sqrt{\frac{(n-1)(n-2)}{R_{\sigma_i}}}, 
 \]
keeping in mind that the scalar curvature $R_{\sigma_i}$ is strictly positive.

Now we define a ``charge''
\begin{align}\label{qidef}
 q_i &\defeq \sqrt{\frac{2}{(n-1)(n-2)}}\frac{\nu(\Psi)_i}{N_i}r_i^{n-1} 
\end{align}
as well as a ``mass'' 
\begin{align}\label{defmi}
 m_i\defeq \frac{r_i^{n-2}}{n}+\frac{(n-1)q_i^2}{n r_i^{n-2}}.
\end{align}

 We need to show for later use that $m_i^2>q_i^2$. To this end, we calculate (plugging in Definition~\ref{defmi} for $m_i$)
 \begin{align*}
  r_i^{2(n-2)}n^2\left(m_i^2-q_i^2\right)= r_i^{4(n-2)}-(n^2-2n+2)q_i^2r_i^{2(n-2)}+(n-1)^2q_i^4 \end{align*}
and this quantity is positive provided that 
 
\[ r_i^{2(n-2)}>(n-1)^2q_i^2.
\]

Plugging in Definition~\ref{qidef} for $q_i$ and then using Equation~\eqref{scalar2} to substitute for $\frac{\nu(\Psi)^2}{N_i^2}$, this is seen to be equivalent to 

\[
 \frac{H_i^2}{R_{\sigma_i}}>\frac{n-2}{n-1},
\]

which is exactly the subextremality condition. This shows that $m_i^2>q_i^2$. 

We may now define 
\[ I_i \defeq \left[a_i, b_i\right]\defeq\left[\left(m_i+\sqrt{m_i^2-q_i^2}\right)^{\frac{1}{n-2}},r_i=\left(\frac{ m_i n}{2}+\frac 1 {2} \sqrt{m_i^2n^2-4(n-1)q_i^2}\right)^\frac{1}{n-2}\right]\]

and set 
 \[\gamma_i \defeq \frac{1}{N_{m_i, q_i}(r)^2}dr^2+\frac{r^2}{r_i^2}\sigma_i, \]

and recall  $N_{m_i, q_i}(r) = \sqrt{1-\frac{2m_i}{r^{n-2}}+\frac{q_i^2}{r^{2(n-2)}}}$ and  $\Psi_{q_i}(r)=\frac{q_i}{\widehat C r^{n-2}}$.

Because of its importance for the subsequent arguments, we state the following fact as a lemma: \newpage
\begin{lem}
Let $\alpha_i, \beta_i>0$ be constants. 
The system $(I_i\times \Sigma_i^{n-1}, \gamma_i, \alpha_i N_{m_i, q_i}, \alpha_i \Psi_{q_i}+\beta_i)$ is a traced-electrostatic system. Furthermore, $\mathbb R\times \{b_i\}\times \Sigma_i^{n-1}$ is a photon sphere in $\left(\mathbb R\times I_i\times\Sigma_i^{n-1}, - (\alpha_i N_{m_i, q_i})^2dt^2+\gamma_i\right)$ and $\{a_i\}\times \Sigma_i^{n-1}$ is a nondegenerate static horizon. 
\end{lem}

\begin{proof}
For the system system $(I_i\times \Sigma_i^{n-1}, \gamma_i, N_{m_i, q_i}, \Psi_{q_i})$, Equations~\eqref{EEVE1}, \eqref{EEVE2}, and~\eqref{EEVEtrace} can be verified by a straightforward computation involving the Christoffel symbols of the metric $\gamma_i$ and a comparison with those of the Reissner--Nordstr\"om manifold of mass $m_i$ and charge $q_i$. 
Since Equations~\eqref{EEVE1},~\eqref{EEVE2} , and~\eqref{EEVEtrace} are invariant under scaling $\left(N, \Psi\right)\mapsto\left(\alpha N, \alpha \Psi +\beta\right)$ of a lapse $N$ and a potential $\Psi$ by positive constants $\alpha, \beta$, the system
$(I_i\times \Sigma_i^{n-1}, \gamma_i, \alpha_i N_{m_i, q_i}, \alpha_i\Psi_{q_i}+\beta_i)$ is also a traced-electrostatic system. 

One verifies by direct computations that $\{b_i\}\times \Sigma_i^{n-1}$ is a photon sphere and $\{a_i\}\times \Sigma_i^{n-1}$ a nondegenerate static horizon. 
\end{proof}

We now choose 
\begin{align*}
 \alpha_i &\defeq \frac{N_i}{N_{m_i, q_i}(r_i)} >0,\\
 \beta_i &\defeq \Psi_i -\alpha_i \frac{q_i}{\widehat C r_i^{n-2}}.
\end{align*}

 Now we combine $(M^n, g, N, \Psi)$ with $(I_i\times \Sigma_i^{n-1}, \gamma_i, \alpha_i N_{m_i, q_i}, \alpha_i\Psi_{q_i}+\beta_i)$ to a new system $(\widetilde M^n, \widetilde g, \widetilde N,\widetilde \Psi)$ by gluing along the boundary components $\Sigma_i^{n-1}$ and setting 
 
 \[
 \begin{aligned}
  \widetilde g& \defeq \begin{cases} g \phantom{blablabla} & \text{ on } M^n,
                    \\
                   \gamma_i & \text{ on } \Sigma_i^{n-1},
                     \end{cases}\\
\widetilde N & \defeq \begin{cases} N & \text{ on } M^n,
                   \\
                   \alpha_i N_{m_i, q_i} \phantom{lalelu} & \text{ on } \Sigma_i^{n-1},
                     \end{cases}\\
\widetilde \Psi & \defeq \begin{cases} \Psi & \text{ on } M^n,
                   \\
                   \alpha_i \Psi_{q_i}+\beta_i \phantom{l} & \text{ on } \Sigma_i^{n-1}.
                     \end{cases}
 \end{aligned}
 \]

 \emph{We proceed to show that $\widetilde g, \widetilde N$, and $\widetilde \Psi$ are well-defined and $C^{1,1}$ across all gluing surfaces $\Sigma_i^{n-1}$.}
 
We intend to use $\widetilde N$ as a smooth collar function across the gluing surface $\Sigma_i^{n-1}$; to this end, we first collect some facts about $\widetilde N$ and $\widetilde \Psi$.

 By the choice of the scaling constants $\alpha_i$ and $\beta_i$, both $\widetilde N$ and $\widetilde \Psi$ are well-defined and continuous across $\Sigma_i^{n-1}$. 
 
 On the side of the glued-in necks $I_i\times \Sigma_i^{n-1}$, the unit normal to $\Sigma_i^{n-1}$ is given as $\widetilde\nu=N_{m_i, q_i}(r_i)\partial_r$. 
 
We use the explicit form of $\widetilde \Psi$ on the necks and the definition of $q_i$ to calculate that 
\[\widetilde\nu(\widetilde\Psi)=\alpha_i N_{m_i, q_i}(r_i)\partial_r (\Psi_{q_i})=-\alpha_i N_{m_i, q_i}(r_i)\frac{q_i}{\widehat C}(n-2)r_i^{-n+1}=\nu(\Psi)_i
\] 

on $I_i\times \Sigma_i^{n-1}$; therefore, the normal derivative of $\widetilde\Psi$ is the same on both sides.

Note that $R_{\sigma_i}$ agrees with the scalar curvature of the photon sphere in the Reissner--Nordstr\"om manifold of mass $m_i$ and charge $q_i$. This can be seen by solving the definition of $m_i$ 
for $r_i^{n-2}=\frac 1 2 m_i n+\frac 1 2 \sqrt{m_i^2n^2-4(n-1)q_i^2}$, where one  also needs to use subextremality $r_i^{2(n-2)}>(n-1)^2 q_i^2$ to rule out the possibility 
 $r_i^{n-2}=\frac 1 2 m_i n-\frac 1 2 \sqrt{m_i^2n^2-4(n-1)q_i^2}$. Plugging this into the definition of $r_i$ and solving for $R_{\sigma_i}$ then gives the term for $R_{\sigma_i}$ which is exactly the scalar curvature of the induced metric of the Reissner--Nordstr\"om photon sphere of mass $m_i$ and $q_i$. 

Since $R_{\sigma_i}$ agrees with the respective Reissner--Nordstr\"om term, one sees from by definition of $q_i$ that $\nu(\widetilde\Psi)$ (which we already showed to agree from both sides) also agrees with the respective value for Reissner--Nordstr\"om. 

Now, by Equation~\eqref{scalar2}, the square of the mean curvature of $\Sigma_i^{n-1}$ on the original side is the same as the square of the mean curvature of the Reissner--Nordstr\"om photon sphere with mass $m_i$ and $q_i$. But since the sign of the mean curvature on the original side is positive by Proposition~\ref{meanconvex}), the mean curvature agrees with the mean curvature of 
the photon sphere in the Reissner--Nordstr\"om manifold of mass $m_i$ and charge $q_i$ (both with respect to the outward pointing unit normal), which is positive. 

On the newly glued-in side, it can be verified by a direct comparison with the Reissner--Nordstr\"om manifold of mass $m_i$ and charge $q_i$ that the mean curvature of $\Sigma_i^{n-1}$ agrees with the one from Reissner--Nordstr\"om. Therefore, the mean curvature of $\Sigma_i^{n-1}$ agrees from both sides.

We proceed to show that the normal derivative of $\widetilde N$ agrees from both sides. 

On the original side $M^n$, it can be expressed via Equation~\eqref{NandH} as 
\[ 
\nu( N)_i =(n-1)H_i N_i.
\]

On the side of the glued-in necks, we compare $\widetilde \nu(\widetilde N)_i$ with the normal derivative of the lapse in the the Reissner--Nordstr\"om manifold of mass $m_i$ and charge $q_i$ at the photon sphere. 
Denoting the outward pointing unit normal to the photon sphere in this Reissner--Nordstr\"om manifold by $\nu_{m_i, q_i}$, we explicitly calculate for the neck that 
\[
\widetilde \nu(\widetilde N)=\alpha_i \nu_{m_i, q_i}\left(N_{m_i, q_i}\right). 
\]

Since for the photon sphere of the Reissner--Nordstr\"om manifold Equation~\eqref{scalar2} holds and we already know that $H_i$ agrees with the the mean curvature $H_{m_i, q_i}$ of 
the photon sphere in the Reissner--Nordstr\"om manifold of mass $m_i$ and charge $q_i$, 
we get on the glued-in side that
\[ \widetilde \nu(\widetilde N)=\alpha_i \nu_{m_i, q_i}\left(N_{m_i, q_i}\right)= \alpha_i (n-1)H_{m_i, q_i} N_{m_i, q_i}=(n-1)H_i\widetilde N_i\] 
agrees with $\widetilde \nu(\widetilde N)$ on the original side. 

Now, as $\widetilde N$ is well-defined, constant on $\Sigma_i^{n-1}$, and its normal derivatives do not vanish and agree from both sides, we can use $\widetilde N$ as a smooth collar function in a neighborhood of  $\Sigma_i^{n-1}$. 
This finally shows that $\widetilde M^n$ is a smooth manifold. 

Since we also showed along the way that $\widetilde \Psi$ is well-defined and its normal derivatives agree from both sides (and since it is smooth away from the gluing surfaces),
$\widetilde \Psi$ is indeed $C^{1,1}$ across $\Sigma_i^{n-1}$. 

Only the regularity of the metric $\widetilde g$ remains to be proven. To this effect, let $\{y^A\}$ be local coordinates on $\Sigma_i^{n-1}$ and flow them
to a neighborhood of $\Sigma_i^{n-1}$ in $\widetilde M^n$ along the level set flow of $\widetilde N$.  We will
show that the components $\widetilde g_{\widetilde N \widetilde N}$, $\widetilde g_{\widetilde N A}$, and $\widetilde g_{AB}$ are $C^{1,1}$ \newpage \noindent across $\Sigma_i^{n-1}$ with respect to the coordinates $\left(\widetilde N, y^A\right)$ 
for all $A,B=1, \dots, n-1$. This is done exactly as in~\cite{cedergalelectro}, so we will be brief:

Continuity and smoothness in tangential directions of $\widetilde g$ in the chosen coordinates are immediate by construction of $\widetilde g$. 
The components $\widetilde g_{\widetilde N A}$ vanish in a neighborhood of $\Sigma_i^{n-1}$ (for each $A=1, \dots, n-1$). 
Hence, we only need to consider the normal derivatives of $\widetilde g_{\widetilde N\widetilde N}$ and  $\widetilde g_{AB}$ for $A,B=1, \dots, n-1$. 

Now
\[
\widetilde g_{AB, \widetilde N}=\frac{2}{\widetilde\nu(\widetilde N)}\widetilde h_{AB}
\]
holds on $\Sigma_i^{n-1}$, where $\widetilde h_{AB}$ denotes the second fundamental form of $\Sigma_i^{n-1}$ in $(\widetilde M^n, \widetilde g)$. By umbilicity and the fact that the mean curvature agrees from both sides, 
$ \partial_{\widetilde N}\left(\widetilde g_{AB}\right)$ is the same from both sides. 

Also, 
\[
 \widetilde g_{ \widetilde N \widetilde N, \widetilde N}=-2\widetilde\nu(\widetilde N)^2\nabla^2 \widetilde N(\widetilde\nu,\widetilde \nu), 
\] 
and by Equation~\eqref{decompolaplacian}, constancy of $\widetilde N$ and $\widetilde\Psi$ on $\Sigma_i^{n-1}$, and Equation~\eqref{EEVE1}, one gets
\[
  \nabla^2 \widetilde N(\widetilde\nu,\widetilde \nu)= \frac{\widehat C^2}{\widetilde N}\widetilde\nu(\widetilde \Psi)^2-H_i\widetilde\nu(\widetilde N),
\]
so that $\widetilde g_{ \widetilde N \widetilde N, \widetilde N}$ agrees from both sides.

Summing up, the system $\left(\widetilde M^n, \widetilde g, \widetilde N, \widetilde \Psi\right)$ we constructed is $C^{1,1}$ on a finite set of hypersurfaces and smooth elsewhere and is at least pseudo-electrostatic. Furthermore, its boundary consists of nondegenerate static horizons. 

\subsection{Doubling} 

Like the authors of~\cite{carlahigh, cedergalelectro} and following the original models for this procedure in~\cite{bunting, ruback}, in this section we double the Riemannian manifold $\widetilde M^n$ that we constructed in the previous section and
glue the two copies along their shared boundary. The metric, the lapse, and the potential will be extended to all of $\widehat M^n$.

As the arguments mirror those of~\cite{carlahigh} and others, we will just briefly sketch them and show that they carry over to our situation with only slight modifications.  

First, we rename $\left(\widetilde M^n, \widetilde g, \widetilde N,\widetilde \Psi\right)$ to $\left(\widetilde M^n_+, \widetilde g_+, \widetilde N_+,\widetilde \Psi_+\right)$, reflect $\widetilde M^n$ as well as 
$\widetilde g, \widetilde N$ and $\widetilde \Psi$ through the boundary $\partial \widetilde M^n$ to obtain a new system that we call $\left(\widetilde M^n_-, \widetilde g_-, \widetilde N_-,\widetilde \Psi_-\right)$. 
Then we glue $\widetilde M^n_+$ and $\widetilde M^n_-$ along their shared boundary and name the resulting manifold $\widehat M^n$. We also set 

\[
 \begin{aligned}
  \widehat g& \defeq \begin{cases} \widetilde g_+ & \text{ on } \widetilde M^n_+,\\
                    \widetilde g_-\phantom{bll} &\text{ on }\widetilde M^n_-,
                     \end{cases}\\
\widehat N & \defeq \begin{cases} \widetilde N_+ & \text{ on }\widetilde M^n_+,
                   \\
                  - \widetilde N_- & \text{ on }\widetilde M^n_-,
                     \end{cases}\\      
 \widehat \Psi & \defeq \begin{cases} \widetilde \Psi_+ & \text{ on }\widetilde M^n_+,
                   \\
                  \widetilde \Psi_- \phantom{bl}& \text{ on }\widetilde M^n_-.
                     \end{cases}
 \end{aligned}
 \]

We will denote the connected components of the gluing surface $\partial \widetilde M^n\subseteq \widehat M^n$ by $\widehat \Sigma_i^{n-1}$, that is, 
\[
\partial \widetilde M^n=\bigcup\limits_{i=1}^l\widehat \Sigma_i^{n-1}.
\]

Each boundary component $\widehat \Sigma_i^{n-1}$ ($1\leq i\leq L$) is a nondegenerate static horizon, either by the construction in the previous subsection (for $1\leq i\leq l$), or by the assumptions of Theorem~\ref{meta} (for $l+1\leq i\leq L$). We fix an $1\leq i\leq L$ to show $C^{1,1}$-regularity of $\left(\widehat M^n, \widehat g, \widehat N ,\widehat \Psi\right)$ across $\partial \widetilde M^n$.

By its construction as an odd function, $\widehat N$ is smooth across $\partial \widetilde M^n$. This allows us to use $\widehat N$ as a smooth collar function across $\partial \widetilde M^n$, showing that $\widehat M^n$ is a smooth manifold. 

The fact that $d\Psi\restriction_{\widehat \Sigma_i^{n-1}}=0$ from Lemma~\ref{horizon} gives at once that $\widehat \Psi$ is $C^{1,1}$ across $\partial \widetilde M^n$. 

We imitate closely the argumentation of~\cite{carlahigh} to show that the metric is of regularity $C^{1,1}$ across the gluing surfaces. To this end, we switch to adapted coordinates $(\widehat N, y^A)$ in a neighborhood of $\widehat \Sigma_i^{n-1}$. It is immediate that
\begin{align*}
\partial_{\widehat N}\left(g_{A\widehat N}\right)&=0
\end{align*}

for all $A,B=2, \dots, n$. 

Denoting by $\widehat \nu_+$ the unit normal to $\partial \widetilde M^n$ pointing into $\widetilde M^n_+$, the level set flow equations give 
\[
\partial_{\widehat N}\left(g_{\widehat N \widehat N}\right)=-2\left(\widehat \nu_+(\widehat N)\right)^2\widehat \nabla^2 \widehat N(\widehat \nu_+, \widehat\nu_+)
\] 
on $\partial\widetilde  M^n$. 

By Formula~\eqref{decompolaplacian}, this reduces to 
\[
\partial_{\widehat N}\left(g_{\widehat N \widehat N}\right)= \Delta \widehat N, 
\]
where we also made use of the facts that  $\widehat N\restriction_{\widehat\Sigma_i^{n-1}}=0$ and that $\widehat \Sigma_i^{n-1}$ has vanishing mean curvature.   

Jointly with Equation~\eqref{EEVE1} and Lemma~\ref{horizon}, this gives 
\[
 \partial_{\widehat N}\left(g_{\widehat N \widehat N}\right)=0
\]
on $\widehat \Sigma_i^{n-1}$. Since the same holds on the other side of $\Sigma_i^{n-1}$, this shows that $g_{\widehat N \widehat N}$ is $C^{1,1}$ across $\Sigma_i^{n-1}$. 

To see that $g_{AB}$ is $C^{1,1}$, one calculates that 
\begin{equation*}
\partial_{\widehat N}\left(g_{AB}\right)=\frac{2}{\nu_+(\widehat N)}h_{AB}
\end{equation*}
for all $A,B=2, \dots, n$, where $h$ is the second fundamental form of $\widehat \Sigma_i^{n-1}$, which vanishes since $\widehat \Sigma_i^{n-1}$ is a static horizon from both sides; so that
 
\begin{equation*}
\partial_{\widehat N}\left(g_{AB}\right)=0
\end{equation*}
for all $A,B=2$. 

Summing up the results of the last two sections, we have constructed a system $\left(\widehat M^n, \widehat g, \widehat N ,\widehat \Psi\right)$ with an ``upper half'' $\widetilde M^n_+$ and a ``lower half'' $\widetilde M^n_-$, 
which is smooth except possibly on a finite collection of hypersurfaces where it is $C^{1,1}$, and such that $(M^n, g)$ 
embeds isometrically into the upper half of $\left(\widehat M^n, \widehat g\right)$, and $\widehat N\restriction_{M^n}=N ,\widehat \Psi\restriction_{M^n}=\Psi$. 

Furthermore, $\widehat g$, $\pm \widehat N$, and $\widehat \Psi$ fulfill Equations~\eqref{EEVE1}--\eqref{EEVE2} and the inequality~\eqref{EEVEineq} on $\widetilde M^n_\pm$ 
(except possibly on the gluing surfaces, where second derivatives might not exist), and $\left(\widetilde M^n_\pm, \widehat g,\pm \widehat N ,\widehat \Psi\right)$ are asymptotically Reissner--Nordstr\"om of mass $m$ and charge $q$.

\subsection{Conformal transformation and applying the positive mass theorem}

In this step, the Riemannian manifold $\left(\widehat M^n, \widehat g\right)$ that was constructed in the previous step will 
turn out to be conformally equivalent to $\left(\mathbb R^n, \delta\right)$ by a conformal factor that is constructed from the functions $\widehat N$ and $\widehat \Psi$. 

We define
\[
\Omega\defeq  \left( \frac {(1+ \widehat N)^2-\widehat C^2\widehat\Psi^2} 4 \right)^{1/(n-2)}.
\]

Note that $\Omega$ is smooth everywhere on $\widehat M^n$, except possibly on a finite collection of hypersurfaces, where it is $C^{1,1}$. 

We need to show that $\Omega$ is positive everywhere.

Defining 
 \[F_\pm\defeq  \widehat{N}-1\pm\widehat C \widehat \Psi\]

on $\widehat M^n$, we know from Lemma~\ref{mandq} that $F_\pm<0$ on the original manifold $  M^n$. Hence, 

\[0<F_+ F_-=(1- N)^2-\widehat C^2 \Psi^2<(1+ N)^2-\widehat C^2 \Psi^2=4 \Omega^{n-2} \text{ on }  M^n. \]
 
On the mirrored image of $M^n$ in the lower half $\widetilde M_-^n$, we can now conclude that 
\[0<(1+ \widehat N)^2-\widehat C^2\widehat\Psi^2 =4 \Omega^{n-2}, \]

using that in this region $\widehat N=-N$ and that we just showed that 
$(1- N)^2-\widehat C^2 \Psi^2>0$.

On the glued-in necks, we can apply a similar trick; since we already know that $F_\pm<0$ on $\Sigma_i^{n-1}$, it suffices to check 
(using the explicit form of $F_\pm$ on the necks) that $F_\pm<0$ on $\widehat \Sigma_i^{n-1}$ and again apply a maximum principle to $F_\pm$ (recalling the PDE for $F_\pm$ given in Equation~\eqref{pde}). Summing up, $\Omega>0$ on all of $\widehat M^n$.

We immediately see that the assumptions of Propositions~\ref{zero} and~\ref{onepoint} are met and may thus conclude that by Proposition~\ref{zero}, the upper half $\widetilde M^n$ with the metric $\Omega^2 \widehat g$ is asymptotically Reissner--Nordstr\"om with mass $0$ and charge $0$. 

To the lower half $\widetilde M^n_-\subseteq \widehat M^n$, we apply Proposition~\ref{onepoint} to insert a point $p_\infty$ into  $(\widehat M^n, \Omega^2 \widehat g)$ such that the resulting manifold $(\widehat M^n_\infty, g_\infty)$ has $C^{1,1}$-regularity at $p_\infty$. 

The scalar curvature of a traced-electrostatic system 
after the conformal transformation we performed 
was calculated in~\cite{bubble} (using Equations~\eqref{EEVE1},~\eqref{EEVE2}, and~\eqref{EEVEtrace}) as
\begin{equation}\label{scalarnonneg}
 \frac{1}{8\widehat N^2\Omega^{2(n-3)}}\left|2\widehat N\widehat \Psi\nabla\widehat  N-\left(\widehat N^2-1+\widehat C^2\widehat \Psi^2\right)\nabla\widehat \Psi\right|^2\geq 0.
\end{equation}

It is easy to check (using the standard formula for the conformally transformed scalar curvature and the just mentioned calculation in~\cite{bubble}) that for the pseudo-electrostatic system $(\widehat M^n, \Omega^2 \widehat g)$ the transformed scalar curvature is bounded from below by the left-hand side expression in~\eqref{scalarnonneg} and therefore also nonnegative. 

To sum up, we have constructed a geodesically complete manifold $(\widehat M^n_\infty, g_\infty)$  with nonnegative scalar curvature and vanishing mass. 
The positive mass theorem for smooth manifolds of arbitrary dimensions was proven in~\cite{SYn}. 
A Ricci flow argument (which is independent of the dimension) in~\cite{MSz} shows that if the rigidity case of the positive mass theorem holds in a certain class of smooth manifolds, 
then it also holds for lower regularity ``manifolds with corners'' in that same class (and the isomorphism to the Euclidean space is smooth wherever the metric is smooth).
In particular, the authors of ~\cite{MSz} cover the case that the metric is only $C^{1,1}$-regular on a finite collection of hypersurfaces. 
Summing up, we have the rigidity statement of the positive mass theorem in arbitrary dimensions for smooth manifolds with $C^{1,1}$-regular hypersurfaces at our disposition. 
We thereby get that $(\widehat M^n_\infty, g_\infty)$ is isometric to the Euclidean space $\left(\mathbb R^n, \delta\right)$, 
and the isometry is smooth except possibly on the lower regularity submanifolds (see also~\cite{carlahigh} for the application of the positive mass theorem to an analogous situation).

\subsection{Recovering the Reissner--Nordstr\"om manifold}

In a last step, we show that the original Riemannian manifold $(M^n, g)$ must have been a piece of the $n$-dimensional spatial Reissner--Nordstr\"om manifold. 
We will not explicitly denote the isometry $(\widehat M^n_\infty, g_\infty) \approx \left(\mathbb R^n, \delta\right)$ in what follows. 

Just like the author of \cite{carlahigh}, we first recall that each boundary component $ \Sigma_i^{n-1}$ is a closed, totally umbilical hypersurface of $(M^n, g)$; 
and hence (umbilicity being invariant under the conformal transformation $g\mapsto \Omega^2 g=\delta$ ) a closed, totally umbilical hypersurface of the Euclidean space. 
As a consequence, each $\Sigma_i^{n-1}$ is a round sphere in $(\mathbb R^n, \delta)$, and thus---the conformal
factor being constant on each boundary component---each $\left(\Sigma_i^{n-1}, g\restriction_{\Sigma_i^{n-1}} \right)$ is an intrinsically round sphere.

Second (and again as in~\cite{carlahigh}), since $\widehat M^n_\infty$ is homeomorphic to $\mathbb R^n$, the doubled manifold $\widehat M^n$ (without the inserted point) must be homeomorphic to $\mathbb R^n\setminus\{0\}$. 
In particular, the $(n-1)$-th fundamental group of $\widehat M^n$ is 
\begin{equation*}
\pi_{n-1}\left(\widehat M^n\right)= \pi_{n-1}\left(\mathbb R^n\setminus\{0\}\right)=\mathbb Z.
\end{equation*}

Recalling that each boundary component $\Sigma_i^{n-1}$ of $M^n$ is a topological sphere, we may now conclude that $\partial M^n$ has only one component. This allows us to drop from now on the index $i$. 

In the last step, we will determine the conformal factor $\Omega^2$, applying a maximum principle. 
Since we know that the conformally transformed manifold $(\widehat M^n_\infty, g_\infty)$ is flat and in particular has vanishing scalar curvature, Inequality~\eqref{scalarnonneg} (for the conformally transformed scalar curvature) reduces to the equality
\begin{equation}\label{temp1}
 2\widehat N\widehat \Psi\nabla\widehat  N=\left(\widehat N^2-1+\widehat C^2\widehat \Psi^2\right)\nabla\widehat \Psi.
\end{equation}

It was computed in \cite{bubble} as a consequence of Equations~\eqref{temp1},~\eqref{EEVE1}, and~\eqref{EEVE2} that the functions 
\begin{align} \label{vpm1}
v_\pm\defeq(1+\widetilde N\pm \widehat C\widetilde\Psi)^{-1}: \widetilde M^n_+\rightarrow\mathbb R
\end{align}
are harmonic with respect to the conformally changed metric $\Omega^2g=\delta$ on $\widetilde M^n_+\subseteq \mathbb R^n$.

The asymptotic conditions from Definition~\ref{asyrnfct} imply that $v_\pm (y)\rightarrow\frac 1 2 $ as $|y|\rightarrow \infty$. 

By the same arguments as above for $\Sigma^{n-1}$, the horizon $\widehat \Sigma^{n-1}$ is a round sphere $\mathbb S_{T}(a)\subseteq \mathbb R^n$ with radius $T$ and center $a\in \mathbb R^n$, and 
$v_\pm\restriction_{\widehat \Sigma^{n-1}}=v_\pm\restriction_{\mathbb S_{T}(a)}$ are constants. We can exclude $v_\pm\restriction_{\mathbb S_{T}(a)}\in\{0, \pm\infty\}$ because 
$v_+v_-=\frac 1 4 \Omega^{-(n-2)}$ on all of $\widetilde M^n$ and $\Omega^{-(n-2)}$ vanishes nowhere. 
Thus, by a maximum principle for elliptic PDEs, the functions $v_\pm$ are uniquely determined as the guessed solutions that are known from Reissner--Nordstr\"om, namely 
\begin{equation}\label{vpm2}
 v_\pm (y)=\left(1+ \left(1-\frac{T^{2(n-2)}+c_\pm^2}{T^{n-2}|y-a|^{n-2}}+\frac{   c_\pm^2}{|y-a|^{2(n-2)}}\right)^{\frac 1 2 }
\pm\frac{   c_\pm}{|y-a|^{n-2}}\right)^{-1}, 
\end{equation}

where $c_\pm$ are constants that are determined by the constants $v_\pm\restriction_{\mathbb S_{T}(a)}$ via $v_\pm\restriction_{\mathbb S_{T}(a)}=\left(1\pm \frac{c_\pm}{T^{n-2}}\right)^{-1}$. 

Now, adding the two equations
\begin{equation*}
1\pm \frac{c_\pm}{T^{n-2}}=\frac{1}{v_\pm\restriction_{\mathbb S_T(a)}}
=1\pm \widehat C\Psi\restriction_{\mathbb S_T(a)}, 
\end{equation*}

leads to $c_+=c_-\eqdef c$.

To determine the constant $c$, we compare the asymptotic behavior of $v_\pm$ that we get from Equations~\eqref{vpm1} and~\eqref{vpm2} and conclude that
\begin{equation*}
 m\pm q=\frac{T^{n-2}}{2}+\frac{c^2}{2T^{n-2}} \pm  c. 
\end{equation*}

This is equivalent to 
\begin{align*}
2m&=T^{n-2}+\frac{c^2}{T^{n-2}},\\
q&=c, 
\end{align*}

so that 
\[
v_\pm=\left(1-N_{m,q}\pm \widehat C\Psi_q \right)^{-1}.
\]
We have now determined $v_\pm$ and hence also $N$, $\Psi$, and the conformal factor $\Omega^2$ as the respective functions known from the Reissner--Nordstr\"om manifold with mass $m$ and charge $q$, and hence we know that $\left(M^n, g\right)$ is exactly the Reissner--Nordstr\"om manifold of mass $m$ and charge $q$. 
The inequality $m>|q|$ was already proven in Lemma~\ref{mandq}, so that now the proof of Theorem~\ref{meta} and thereby also of Theorem~\ref{main} is complete. 

\section{Acknowledgements}

I would like to thank Carla Cederbaum for suggesting this problem and for helpful discussions. 

This work is supported by the Institutional Strategy of the University of T\"ubingen (Deutsche Forschungsgemeinschaft, ZUK 63).
\bibliography{uniqueness}

\bibliographystyle{plain}
\end{document}